\documentclass[10pt]{article}
\usepackage{graphicx}
\usepackage[utf8]{inputenc}
\usepackage{amsmath,amsfonts,amssymb,dsfont}
\usepackage{tikz}
\usepackage{pgfplots, pgf}
\usepackage{authblk}
\usepackage{marvosym}
\usepackage[labelfont=bf]{caption}
\usepackage{subcaption} 
\usepackage{todonotes} 
\usepackage{url}
\usepackage{comment}
\usepackage{mathdots}

\usepackage{palatino}   
\usepackage{euler}      

\usetikzlibrary{arrows}

      \newtheorem{theorem}{Theorem} 
      \newtheorem{corollary}{Corollary}
      \newtheorem{proposition}{Proposition}
      \newtheorem{definition}{Definition}

\newenvironment{proof}{{\sc Proof:}}{     
~\hfill\rule{2mm}{3mm}\vspace{.1in}}
\title{Distributions of Matching Distances in Topological Data Analysis}
\author{So Mang Han, Taylor Okonek, Nikesh Yadav, Xiaojun Zheng}
\affil{St.\ Olaf College, Northfield, Minnesota, USA}
\date{\today}

\begin{document}
\maketitle

\begin{abstract}
Topological data analysis seeks to discern topological and geometric structure of data, and to understand whether or not certain features of data are significant as opposed to random noise. While progress has been made on statistical techniques for single-parameter persistence, the case of two-parameter persistence, which is highly desirable for real-world applications, has been less studied. This paper provides an accessible introduction to two-parameter persistent homology and presents results about matching distance between 2-parameter persistence modules obtained from families of simple point clouds. Results include observations of how differences in geometric structure of point clouds affect the matching distance between persistence modules. We offer these results as a starting point for the investigation of more complex data.
\end{abstract}

\section{Introduction and Motivation} 

\textbf{Topological data analysis (TDA)} is a collection of methods used to discern the shape of data. TDA detects topological features, such as clusters, holes, and voids.
Topological methods are especially useful for high-dimensional, noisy data.
TDA has been applied in numerous settings, including image analysis \cite{image}, protein structure \cite{protein}, texture representation in images \cite{perea}, astronomical data \cite{astronomical}, and neuroscience \cite{sizemore}.

One of the main tools in TDA is \textbf{persistent homology}. Persistent homology associates to a dataset an algebraic object known as a persistence module, which encodes topological features of the data. The study of persistence modules can then reveal insights about the data that underlies the modules.

One common problem is to compare two datasets via their persistence modules. In this setting, notions of distance between persistence modules are useful for quantifying the amount of difference between persistence modules. This paper examines one such distance, the matching distance, which is easily computed. Our goal is to understand how the matching distance quantifies similarity between datasets.

We computed matching distances between persistence modules arising from datasets of two types. The first type of dataset consists of three points in various configurations. The second type of dataset consists of two circles with radii $r$, and circles were separated by a distance $d$. We examined how changes to $r$ or $d$ affect the matching distance.

The organization of the paper is as follows: In Section 2, we provide mathematical background for persistent homology and the matching distance. In Section 3, we describe our data analysis and matching distance computations. Discussion and directions for future research are provided in Section 4.

\section{Mathematical Background} 
 
Persistent homology is one of the main tools in TDA and can be applied to many types of data, including real-valued functions and sets of points in Euclidean space. It quantifies multi-scale topological features of data: connected components, holes, voids, and their higher-dimensional analogs.


Previous research has used persistent homology to discern topological structure in data from many fields \cite{image, protein, perea, astronomical, sizemore}. 
Nearly all of this prior work has used one-parameter persistent homology, which produces easily-visualized descriptors called barcodes, but which is sensitive to outliers. 
This sensitivity can be avoided by using two-parameter persistence.
We give here a brief introduction to persistent homology in both the one- and two-parameter settings; more detailed surveys of the subject are found in \cite{persistenthomology, ghristbarcodes}.

\subsection{One-parameter persistence}

Given a set of point-cloud data, we first build a simplicial complex.
Our building blocks are \textbf{simplices}: a point is a $0$-simplex; an edge is a $1$-simplex; a triangular face is a $2$-simplex. 
More formally, an $n$-simplex is an $n$-dimensional geometric object that is the convex hull of $n+1$ points which are not contained in any $(n-1)$-dimensional plane.
A \textbf{simplicial complex} $X$ is a set of simplices such that if $v \in X$, then every face of $v$ is also in $X$, and if $v, w \in X$, then $v \cap w$ is also in $X$.
A common type of simplicial complex built from a point cloud is the Rips complex, which we now define.

\begin{definition} Given a collection of points $\{ x_\alpha \}$ in Euclidean space $\mathds{E}^n$ and $\epsilon > 0$, the \textnormal{\textbf{Rips complex}} $\mathcal{R}_\epsilon$ is the simplicial complex whose $k$-simplices correspond to unordered $( k + 1 )$-tuples of points $\{x_\alpha \}^k_0$ whose pairwise distances are at most $\epsilon$. 
\end{definition}

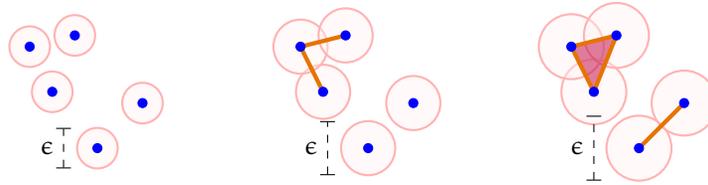
\begin{figure}[h]
  \captionsetup{justification=centering}  
  \begin{center}
  \begin{tikzpicture}[scale=0.3]
    \foreach \Point in {(7,0),(5,2.5),(9,2),(4,4.5),(6,5)}{
      \filldraw[color=red!30, fill=red!5, fill opacity=0.4, thick] \Point circle (0.9);
      \filldraw[blue] \Point circle (0.2);
    }
    \draw [|-|, thin, dashed] (5.5,-0.9) -- node [left] {$\epsilon$} ++(0,1.8);

    \foreach \Point in {(19,0),(17,2.5),(21,2),(16,4.5),(18,5)}{
      \filldraw[color=red!30, fill=red!5, fill opacity=0.4, thick] \Point circle (1.2);
    }
    \draw [orange!90!black, ultra thick] (16,4.5) -- (18,5);
    \draw [orange!90!black, ultra thick] (16,4.5) -- (17,2.5);
    \foreach \Point in {(19,0),(17,2.5),(21,2),(16,4.5),(18,5)}{
      \filldraw[blue] \Point circle (0.2);
    }
    \draw [|-|, thin, dashed] (17.2,-1.2) -- node [left] {$\epsilon$} ++(0,2.4);

    \foreach \Point in {(31,0),(29,2.5),(33,2),(28,4.5),(30,5)}{
      \filldraw[color=red!30, fill=red!5, fill opacity=0.4, thick] \Point circle (1.44);
    }
    \draw [orange!90!black, ultra thick] (31,0) -- (33,2);
    \draw[orange!90!black, ultra thick, fill=purple, fill opacity=0.5] (28,4.5) -- (30,5) -- (29,2.5) -- cycle;
    \foreach \Point in {(31,0),(29,2.5),(33,2),(28,4.5),(30,5)}{
      \filldraw[blue] \Point circle (0.2);
    }
    \draw [|-|, thin, dashed] (29,-1.44) -- node [left] {$\epsilon$} ++(0,2.88);

  \end{tikzpicture}
  \end{center}

  \caption{Three Rips complexes built from a dataset of five points, with different scale parameters $\epsilon$.}
\label{filtration}
\end{figure}

In other words, the Rips complex depends on a scale parameter $\epsilon$. The complex contains an edge between two points if and only if the distance between the points is at most $\epsilon$. The complex contains a triangular face for any three points whose pairwise distances are at most $\epsilon$. 
Figure \ref{filtration} shows three Rips complexes built from the same underlying point cloud, but with different scale parameters. For illustration purposes, we have drawn a ball of diameter $\epsilon$ around each data point. The complex contains an edge for each pair of balls that intersect and a triangular face for each three balls that intersect pairwise.

A Rips complex is built with a fixed scale parameter $\epsilon$, but usually no single choice of $\epsilon$ reveals all structure of the data. Instead, we consider many Rips complexes, one for \emph{every} positive value $\epsilon$. Imagine growing balls of diameter $\epsilon$ centered at each point: as $\epsilon$ increases from zero, an $n$-simplex appears whenever $n+1$ balls pairwise intersect. 
Figure \ref{filtration} shows three snapshots of this process, which leads us to the concept of a filtration.

A \textbf{filtration} is a sequence of simplicial complexes, each a subcomplex of the next: 
\begin{equation}\label{filt_eq}
  X_1 \subset X_2 \subset X_3 \subset \cdots \subset X_n \subset \cdots.
\end{equation}
Figure \ref{barcode} illustrates a filtration. In the figure, $X_i \hookrightarrow X_{i+1}$ denotes a map that takes each simplex in $X_i$ to its corresponding simplex in $X_{i+1}$; this is possible because $X_i$ is a subcomplex of $X_{i+1}$.
Note that if a simplex appears in $X_i$, it must be present in $X_j$ for all $j > i$.

If every complex in a filtration is a Rips complex, then we call the filtration a \textbf{Rips filtration}.
Given a finite point set, simplices appear at only finitely many values of $\epsilon$. Thus, a Rips filtration of a finite point set can be denoted
\[ \mathcal{R}_{\epsilon_0} \subset \mathcal{R}_{\epsilon_1} \subset \mathcal{R}_{\epsilon_2} \subset \cdots \subset \mathcal{R}_{\epsilon_n} \]
for some sequence $0 = \epsilon_0 < \epsilon_1 < \epsilon_2 < \cdots < \epsilon_n$.

Figure \ref{barcode} illustrates a Rips filtration; note that we have not drawn a complex for every  $\epsilon$ at which a simplex appears.
The shaded areas represent the triangular faces and, in $X_6$, the boundary of a $3$-simplex.

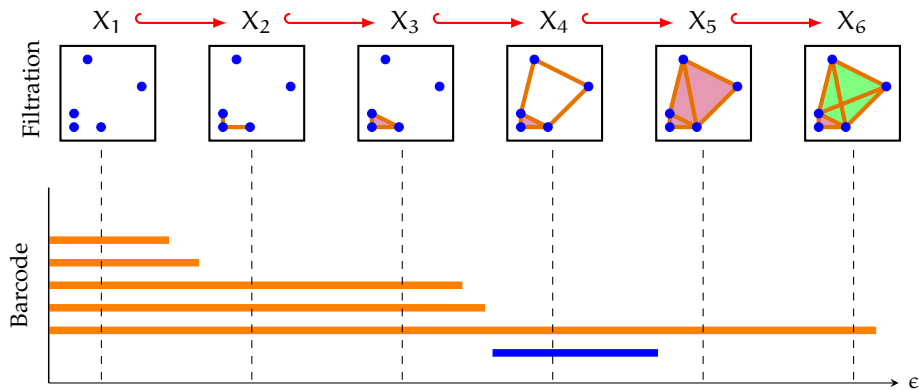
\begin{figure}[bh]
  \centering
  \captionsetup{justification=centering}  

  \begin{tikzpicture}[scale=0.18]
    \node [rotate = 90] at (-12.2,3.8) {\small Filtration};
    
    \draw [thick] (-10,0) rectangle ++(7,7);
    \draw [thick] (1,0) rectangle ++(7,7);
    \draw [thick] (12,0) rectangle ++(7,7);
    \draw [thick] (23,0) rectangle ++(7,7);
    \draw [thick] (34,0) rectangle ++(7,7);
    \draw [thick] (45,0) rectangle ++(7,7);
  
    \filldraw [purple, opacity=0.4] (13,1) -- (13,2) -- (15,1);
    \filldraw [purple, opacity=0.4] (24,1) -- (24,2) -- (26,1);
    \filldraw [purple, opacity=0.4] (35,1) -- (35,2) -- (37,1);
    \filldraw [purple, opacity=0.4] (46,1) -- (46,2) -- (48,1);
    \filldraw [purple, opacity=0.4] (36,6) -- (37,1) -- (35,2);
    \filldraw [purple, opacity=0.4] (36,6) -- (40,4) -- (37,1);
    \filldraw[green,opacity=0.5] (46,2) -- (48,1) -- (51,4) -- (47,6);
    
    \draw [orange!90!black, ultra thick] (2,1) -- (2,2);
    \draw [orange!90!black, ultra thick] (2,1) -- (4,1);
    \draw [orange!90!black, ultra thick] (13,1) -- (13,2);
    \draw [orange!90!black, ultra thick] (13,2) -- (15,1);
    \draw [orange!90!black, ultra thick] (24,1) -- (24,2);
    \draw [orange!90!black, ultra thick] (24,2) -- (26,1);
    \draw [orange!90!black, ultra thick] (35,1) -- (35,2);
    \draw [orange!90!black, ultra thick] (46,1) -- (46,2);
    \draw [orange!90!black, ultra thick] (46,2) -- (48,1);
    \draw [orange!90!black, ultra thick] (13,1) -- (15,1);
    \draw [orange!90!black, ultra thick] (24,1) -- (26,1);
    \draw [orange!90!black, ultra thick] (35,1) -- (37,1);
    \draw [orange!90!black, ultra thick] (46,1) -- (48,1);
    \draw [orange!90!black, ultra thick] (24,2) -- (25,6);
    \draw [orange!90!black, ultra thick] (26,1) -- (29,4);
    \draw [orange!90!black, ultra thick] (25,6) -- (29,4);
    \draw [orange!90!black, ultra thick] (36,6) -- (40,4);
    \draw [orange!90!black, ultra thick] (47,6) -- (51,4);
    \draw [orange!90!black, ultra thick] (37,1) -- (40,4);
    \draw [orange!90!black, ultra thick] (35,2) -- (36,6);
    \draw [orange!90!black, ultra thick] (37,1) -- (35,2);
    \draw [orange!90!black, ultra thick] (36,6) -- (37,1);
    \draw [orange!90!black, ultra thick] (48,1) -- (47,6);
    \draw [orange!90!black, ultra thick] (46,2) -- (47,6);
    \draw [orange!90!black, ultra thick] (46,2) -- (51,4);
    \draw [orange!90!black, ultra thick] (48,1) -- (51,4);
    
    \node at (-6.5,8.8) {$X_1$};
    \node at (4.5,8.8) {$X_2$};
    \node at (15.5,8.8) {$X_3$};
    \node at (26.5,8.8) {$X_4$};
    \node at (37.5,8.8) {$X_5$};
    \node  at (48.5,8.8) {$X_6$};
    
    \draw [thick, right hook-latex, red] (-4.5, 8.7) -- ++(6.8,0);
    \draw [thick, right hook-latex, red] (6.5, 8.7) -- ++(6.8,0);
    \draw [thick, right hook-latex, red] (17.5, 8.7) -- ++(6.8,0);
    \draw [thick, right hook-latex, red] (28.5, 8.7) -- ++(6.8,0);
    \draw [thick, right hook-latex, red] (39.5, 8.7) -- ++(6.8,0);
    
    \foreach \Point in {(-9,1),(-9,2),(-7,1),(-8,6),(-4,4),(2,1),(2,2),(4,1),(3,6),(7,4),(13,1),(13,2),(15,1),(14,6),(18,4),(24,1),(24,2),(26,1),(25,6),(29,4),(35,1),(35,2),(37,1),(36,6),(40,4),(46,1),(46,2),(48,1),(47,6),(51,4)}{
        \filldraw[blue] \Point circle (0.32);
    }
    \end{tikzpicture}

    \begin{tikzpicture}{scale=0.18}

    \draw [line width=0.1cm,orange] (0.01,-4.2) -- (1.6,-4.2);
    \draw [line width=0.1cm,orange] (0.01,-4.5) -- (2,-4.5);
    \draw [line width=0.1cm,orange] (0.01,-4.8) -- (5.5,-4.8);
    \draw [line width=0.1cm,orange] (0.01,-5.1) -- (5.8,-5.1);
    \draw [line width=0.1cm,orange] (0.01,-5.4) -- (11,-5.4);
    \draw [line width=0.1cm,blue] (5.9,-5.7) -- (8.1,-5.7);
    \node [rotate = 90] at (-0.4,-4.8) {\small Barcode};

    \draw [dashed] (0.7,-3) --(0.7,-6.1);
    \draw [dashed] (2.7,-3) --(2.7,-6.1);
    \draw [dashed] (4.7,-3) --(4.7,-6.1);
    \draw [dashed] (6.7,-3) --(6.7,-6.1);
    \draw [dashed] (8.7,-3) -- (8.7,-6.1);
    \draw [dashed] (10.7,-3) -- (10.7,-6.1);

    \draw [->,>=stealth] (0,-6.1) -- (11.3,-6.1);
    \draw [line width=0.02cm] (0,-6.1) -- (0,-3.5);
    \node [scale = 0.03cm] at (11.5,-6.1) {$\epsilon$};
  \end{tikzpicture}

\caption{\textit{Top}: Six complexes are shown from a Rips filtration built from a five-point dataset. \textit{Bottom}: The corresponding barcode; orange bars are zeroth homology (which tracks connected components) and blue bars are first homology (which tracks holes).}
\label{barcode}
\end{figure}

Any topological feature (such as a component or a hole) appears in the filtration at some scale parameter $\epsilon_1$ and disappears at some scale parameter $\epsilon_2$; the pair $(\epsilon_1, \epsilon_2)$ gives the persistence of the feature.
Plotting each persistence pair as a interval along the scale axis produces a \textbf{barcode}, as seen at the bottom of Figure \ref{barcode}.
The orange bars in Figure \ref{barcode} represent components: each of the five points corresponds to an orange bar starting at $\epsilon = 0$. An orange bar ends at each $\epsilon$ value at which two components become connected.
The blue bar represents the hole in the simplicial complex, which persists over a range of $\epsilon$ values.

The information in a barcode can also be visualized as a \textbf{persistence diagram}, which is a collection of points above the diagonal in the $xy$-plane. Bars in the barcode are in one-to-one correspondence with points in the persistence diagram.
A bar from $a$ to $b$ is plotted as the point $(a,b)$ in the persistence diagram.

In order to quantify the topological features of a simplicial complex, as illustrated in the barcode, we use the mathematics of \textbf{homology}. Homology associates a vector space to each simplicial complex and a linear map to each inclusion map in the filtration. 
The homology vector space $H_k(X)$ is generated by the $k$-dimensional holes of simplicial complex $X$.
Making this precise requires some definitions, which we introduce briefly; for more details, see {\cite{larry}}.

Let $C_k$ be a vector space whose basis consists of all $k$-simplices in simplicial complex $X$.
That is, $C_k$ contains \textbf{$k$-chains}, which are sums of $k$-simplices with coefficients in a field $\mathbb{F}$. In TDA, $\mathbb{F}$ is usually chosen to be the $2$-element field, a choice we make in this paper as well.
The \textbf{boundary operator} $\partial_k : C_k \to C_{k-1}$ maps a $k$-simplex to the sum of its $(k-1)$-faces, extending by linearity to $k$-chains. 
Let $B_k \subseteq C_K$ be the subspace of \textbf{boundaries}, which are images of $\partial_{k+1}$. 
Let $Z_k \subseteq C_k$ be the subspace of \textbf{cycles}, defined by the property that $v \in Z_k$ if and only if $\partial_k(v) = 0$.
Crucially, $B_k \subseteq Z_k$, since $\partial \circ \partial = 0$.
We then define the \textbf{homology vector space} $H_k = Z_k / B_k$.
Thus, $H_k$ is a vector space consisting of all cycles that are not boundaries.
The dimension of $H_k$ is the number of equivalence classes of holes, in the sense that two holes are equivalent if they differ by a boundary.


The homology of a filtration is a one-parameter \textbf{persistence module}. The inclusion maps in the filtration induce linear maps between the homology vector spaces. Specifically, the degree $i$ homology of the filtration in equation \eqref{filt_eq} is a persistence module consisting of the following vector spaces and linear maps:
\[ H_i(X_1) \to H_i(X_2) \to H_i(X_3) \to \cdots \to H_i(X_n) \to \cdots. \]
The structure theorem for persistence modules says each persistence module is the sum of interval modules; each interval gives the persistence of one topological feature in the filtration {\cite{Crawley}}. Thus, a barcode is a visualization of a persistence module, which each interval module shown as a bar.

In order to compare barcodes, we need a notion of distance between barcodes. We use the bottleneck distance, which is easily computable, though other options exist \cite{bubenik}. Before defining the bottleneck distance, we introduce the concept of a matching, which we explain in terms of persistence diagrams.

A \textbf{matching} $\eta$ between persistence diagrams $\mathcal{D}_1$ and $\mathcal{D}_2$ pairs each point in $\mathcal{D}_1$ with a point in $\mathcal{D}_2$ or a point on the diagonal line, and pairs each point in $\mathcal{D}_2$ with a point in $\mathcal{D}_1$ or a point on the diagonal.
For an illustration of a matching between two persistence diagrams, see Figure \ref{matching}.
By convention, we use the $L_\infty$ metric to obtain the distance from a point $x=(x_1,x_2)$ to its matched point $\eta(x)=(y_1,y_2)$:
\[ ||x-\eta(x)||_\infty = \max(|x_1-y_1|,|x_2-y_2|). \]
Let the \emph{size} of a matching refer to the supremum of the $L_\infty$ distance between matched points.
Among all possible matchings, we seek a matching with the smallest size.
The bottleneck distance between $\mathcal{D}_1$ and $\mathcal{D}_2$ is the size of this optimal matching, as defined below \cite{persistenthomology}. 

\begin{definition} 
The \textbf{bottleneck distance} between persistence diagrams $\mathcal{D}_1$ and $\mathcal{D}_2$ is
\[ d_B(\mathcal{D}_1, \mathcal{D}_2) = \inf_\eta \sup_x || x - \eta(x) ||_\infty, \]
where the supremum is taken over all matched points $x$ and the infimum is taken over all matchings $\eta$.
\end{definition}

Figure \ref{matching} shows the optimal matching between two persistence diagrams $\mathcal{D}_1$ and $\mathcal{D}_2$.
The size of this matching is given by the max $L_\infty$ distance between matched points, which is $\max(|a-c|,|b-d|)$. Since no other matching between these persistence diagrams has smaller size, the bottleneck distance $d_B(\mathcal{D}_1, \mathcal{D}_2)$ is equal to $\max(|a-c|,|b-d|)$.

\begin{figure}[h]
  \centering
  \captionsetup{justification=centering}  

  \begin{tikzpicture}[scale=0.65]
    \draw [<->,>=stealth] (7,0) -- (0,0) -- (0,7);
    \draw (0,0) -- (7,7);

    \draw [very thick,black] (3,6) -- (3.6,6.5);
    \draw [very thick,black] (2,5) -- (1.9,5.3);
    \draw [very thick,black] (2.4,2.4) -- (2,2.8);
    \draw [very thick,black] (3.3,3.7) -- (3.5,3.5);
    \draw [very thick,black] (5.7,6.4) -- (6.05,6.05);
    \draw [very thick,black] (4.6,5) -- (4.8, 4.8);

    \foreach \Point in {(3,6),(2,5),(4.6,5)}{
        \filldraw [blue] \Point circle (0.1);
    }
    \foreach \Point in {(3.6,6.5),(1.9,5.3),(2,2.8),(3.3,3.7),(5.7,6.4)}{
        \filldraw [red] \Point circle (0.1);
    }

    \node [scale = 0.03cm] at (3,5.6) {$(a,b)$};
    \node [scale = 0.03cm] at (3.6,6.9) {$(c,d)$};
  \end{tikzpicture}

  \caption{A matching between persistence diagrams $\mathcal{D}_1$ (plotted in blue) and $\mathcal{D}_2$ (plotted in red). Some points in each diagram are matched to the diagonal.}
\label{matching}
\end{figure}
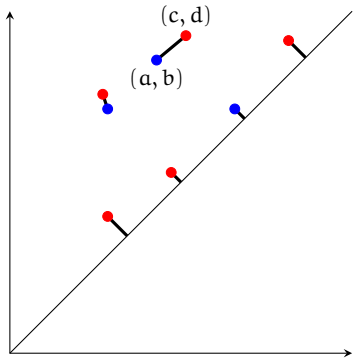


While one-parameter persistence is stable with respect to perturbation of the point cloud data, it is unstable with respect to the presence of outliers.
A density estimator on the points (i.e., a function that indicates whether each point has many nearby neighbors) might be able to identify outliers, but this requires introducing a threshold. Instead, we prefer to use the density estimator as a second filtration parameter, which brings us into the realm of two-parameter persistence.

\subsection{Two-parameter persistence}

Two-parameter persistence arises from data that is simultaneously indexed by two parameters. For example, suppose we have a point cloud $\mathcal{P}$ and a real-valued function $f : \mathcal{P} \to \mathbb{R}$ on each point. In particular, $f$ may arise from a density estimator.
For any $r \in \mathbb{R}$, let 
\[ f^{-1}(-\infty,r] = \{ p \in \mathcal{P} \mid f(p) \le r\}. \]
We can then construct a Rips filtration from $f^{-1}(-\infty,r]$.
Repeating this construction for an increasing sequence $r_1 , r_2, \ldots, r_n$, we obtain a sequence of Rips filtrations, which yields a bifiltration.

A $\textbf{bifiltration}$ is a set of simplicial complexes, each indexed by two parameters, with inclusion maps in the direction of each increasing parameter.
Specifically, the set of simplicial complexes $\{C_{i,j}\}_{i,j}$ forms a bifiltration if there exist commuting inclusion maps 
\[ C_{i,j} \hookrightarrow C_{i',j'} \]
whenever $i \le i'$ and $j \le j'$.
Figure \ref{bifiltration} (left) gives an example of a bifiltration.

The homology of a bifiltration is a 2-parameter persistence module, which is a set of vector spaces $H_p(C_{i,j})$ with commuting linear maps in the directions of increase of $i$ and $j$, as illustrated in Figure \ref{bifiltration} (right). 

\begin{figure}[h]
  \centering
  \captionsetup{justification=centering}  
  \begin{tikzpicture}[scale=0.5]

    \draw [orange!90!black, ultra thick] (1,6.5) -- (1.75,7.5);
    \draw [orange!90!black, ultra thick] (1,11.5) -- (1.75, 12.5);
    \foreach \Point in {(1,1.5), (1.75,2.5), (1,6.5), (1.75,7.5), (1,11.5), (1.75, 12.5)}{
      \filldraw [blue] \Point circle (0.2);
    }
    \draw [orange!90!black, ultra thick] (6,6.5) -- (6.75,7.5);
    \draw[orange!90!black, ultra thick, fill=purple,fill opacity=0.5] (6,11.5) -- (8,11.5) -- (6.75,12.5) -- cycle;

    \foreach \Point in {(6,1.5), (8,1.5), (6.75,2.5), (6,6.5), (8,6.5), (6.75,7.5), (6,11.5), (8,11.5), (6.75, 12.5)}{
      \filldraw [blue] \Point circle (0.2);
    }

    \draw [gray] (0.3,0.9) rectangle ++ (2.2,2.2);
    \draw [gray] (0.3,5.9) rectangle ++ (2.2,2.2);
    \draw [gray] (0.3,10.9) rectangle ++ (2.2,2.2);
    \draw [gray] (5.4,0.9) rectangle ++ (3.3,2.2);
    \draw [gray] (5.4,5.9) rectangle ++ (3.3,2.2);
    \draw [gray] (5.4,10.9) rectangle ++ (3.3,2.2);

    \node at (1.4,3.6) {$C_{1,1}$};
    \node at (1.4,8.6) {$C_{1,2}$};
    \node at (1.4,13.6) {$C_{1,3}$};
    \node at (7.05,3.6) {$C_{2,1}$};
    \node at (7.05,8.6) {$C_{2,2}$};
    \node at (7.05,13.6) {$C_{2,3}$};

    \draw [thick, right hook-latex, red] (2.9, 2) -- ++(2,0);
    \draw [thick, right hook-latex, red] (2.9, 7.1) -- ++(2,0);
    \draw [thick, right hook-latex, red] (2.9, 12.1) -- ++(2,0);
    \draw [thick, right hook-latex, red] (1.4, 4.2) -- ++(0,1.4);
    \draw [thick, right hook-latex, red] (1.4, 9.2) -- ++(0,1.4);
    \draw [thick, right hook-latex, red] (6.95, 4.2) -- ++(0,1.4);
    \draw [thick, right hook-latex, red] (6.95, 9.2) -- ++(0,1.4);
       
  \end{tikzpicture}
  \hspace{2cm}
  \begin{tikzpicture}[scale=0.63]
    \draw[white] (-1,-0.8) rectangle (5,9);

    \node at (0,0) {$H_i(C_{1,1})$};
    \node at (0,4) {$H_i(C_{1,2})$};
    \node at (0,8) {$H_i(C_{1,3})$};
    \node at (4,0) {$H_i(C_{2,1})$};
    \node at (4,4) {$H_i(C_{2,2})$};
    \node at (4,8) {$H_i(C_{2,3})$};

    \draw [->,>=stealth,red,thick] (1.5,0) -- ++(1,0);
    \draw [->,>=stealth,red,thick] (1.5,4) -- ++(1,0);
    \draw [->,>=stealth,red,thick] (1.5,8) -- ++(1,0);
    \draw [->,>=stealth,red,thick] (0,0.8) -- ++(0,2.4);
    \draw [->,>=stealth,red,thick] (4,0.8) -- ++(0,2.4);
    \draw [->,>=stealth,red,thick] (0,4.8) -- ++(0,2.4);
    \draw [->,>=stealth,red,thick] (4,4.8) -- ++(0,2.4);
  \end{tikzpicture}

\caption{A sample bifiltration (left) and its corresponding 2-parameter persistence module in homology degree i (right).}
\label{bifiltration}
\end{figure}
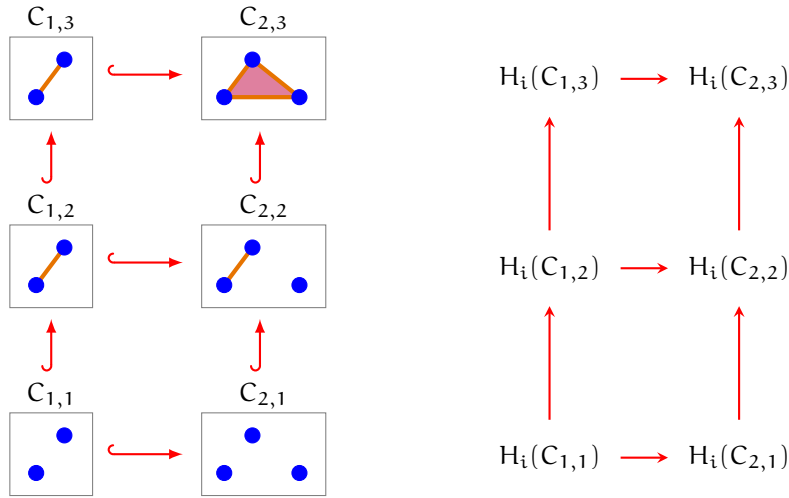

Unfortunately, the algebraic structure of 2-parameter persistence modules is extremely complicated, and there is no reasonable ``barcode'' for such modules \cite{lesnick}.
Instead, we can obtain a barcode along any line with nonnegative slope in the 2-parameter space by restricting the 2-parameter persistence module to such a line, as we now explain.

Let $\mathcal{M}$ be a 2-parameter persistence module with parameter values in discrete indexing sets $I$ and $J$.
Denote the vector spaces in $\mathcal{M}$ as $M_{i,j}$ for every $i \in I$ and $j \in J$, with linear maps $M_{i,j} \to M_{i',j'}$ whenever $i \le i'$ and $j \le j'$.
We may then adopt a continuous perspective, assigning a vector space from $\mathcal{M}$ to every point in $(x,y) \in \mathbb{R}^2$.
If $x < \min(I)$ or $y < \min(J)$, then the point $(x,y)$ is assigned the zero vector space; otherwise, the vector space assigned to $(x,y)$ is $M_{a,b}$, where $a = \max\{i \in I \mid i \le x\}$ and $b = \max\{j \in J \mid j \le y\}$, as illustrated in Figure \ref{2d}.

Let $\ell$ be a line in $\mathbb{R}^2$ with non-negative slope. Let $\mathcal{M}_\ell$ be the 1-parameter persistence module obtained by restricting $\mathcal{M}$ to line $\ell$: every point along $\ell$ is assigned the homology vector space of $\mathcal{M}$ at that point in $\mathbb{R}^2$, with linear maps induced from $\mathcal{M}$ (as in Figure \ref{2d}). 
Since $\mathcal{M}_\ell$ is a 1-parameter persistence module, it has a barcode, or equivalently, a persistence diagram.

\begin{figure}[h]
  \centering
  \captionsetup{justification=centering}  
  \begin{tikzpicture}[scale=1.5]
    
    \filldraw[red!40] (1,1) -- (1,2) -- (2,2) -- (2,1);
    \node at (1.5,1.5) {$M_{1,1}$};
    \node at (2.5,2.5) {$M_{2,2}$};
    \node at (1.5,2.5) {$M_{1,2}$};
    \node at (2.5,1.5) {$M_{2,1}$};
    \node at (1.5,3.5) {$M_{1,3}$};
    \node at (2.5,3.5) {$M_{2,3}$};
    
    \node at (0.5,0.6) {$\iddots$};
    \node at (1.5,0.6) {$\vdots$};
    \node at (2.5,0.6) {$\vdots$};
    \node at (3.5,0.6) {$\ddots$};
    
    \node at (0.5,1.5) {$\cdots$};
    \node at (0.5,2.5) {$\cdots$};
    \node at (0.5,3.5) {$\cdots$};
    \node at (3.5,1.5) {$\cdots$};
    \node at (3.5,2.5) {$\cdots$};
    \node at (3.5,3.5) {$\cdots$};

    \node at (0.5,4.5) {$\ddots$};
    \node at (1.5,4.5) {$\vdots$};
    \node at (2.5,4.5) {$\vdots$};
    \node at (3.5,4.5) {$\iddots$};
    
    \draw [red, thick] (1,0.2) -- (1,4.8);
    \draw [red, ultra thick] (1,1) -- (1,2);
    \draw [red, thick] (2,0.2) -- (2,1);
    \draw [red, ultra thick, dashed] (2,1) -- (2,2);
    \draw [red, thick] (2,2) -- (2,4.8);
    \draw [red, thick] (3,0.2) -- (3,4.8);
    
    \draw [red, thick] (0.2,1) -- (3.8,1);
    \draw [red, thick] (0.2,2) -- (1,2);
    \draw [red, ultra thick] (1,1) -- (2,1);
    \draw [red, ultra thick, dashed] (1,2) -- (2,2);
    \draw [red, thick] (2,2) -- (3.8,2);
    \draw [red, thick] (0.2,3) -- (3.8,3);
    \draw [red, thick] (0.2,4) -- (3.8,4);

    \draw[<->, blue, ultra thick] (0.2,0.7) -- ++(3.6,3.6);
    \node [blue] at (3.5, 3.8) {$\ell$};

    \draw [->,>=stealth] (0.1,0) -- node[below] {Parameter 1} ++(3.8,0);
    \draw [->,>=stealth] (0,0.1) -- ++(0,4.8);
    \node [rotate=90] at (-0.18,2.5) {Parameter 2};
  \end{tikzpicture}

  \caption{A 2-parameter persistence module with discrete parameter values may be viewed from a continuous perspective. Vector space $M_{1,1}$ is associated with all points inside the shaded region (with half-open boundary) in the 2-parameter plane, and similarly for the other vector spaces. The vector spaces along any line $\ell$ with nonnegative slope then form a 1-parameter persistence module, with linear maps induced from the 2-parameter persistence module.} 
\label{2d}
\end{figure}
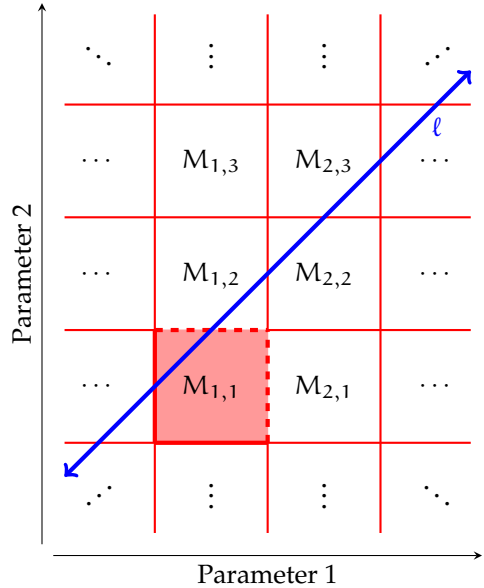

Furthermore, we can define a distance between 2-parameter persistence modules by considering the bottleneck distances between persistence diagrams along all possible lines through the 2-parameter space.
In the following definition, $\mathcal{D}(\mathcal{M}_\ell)$ denotes the persistence diagram of the 1-parameter module $\mathcal{M}_\ell$.

\begin{definition}\label{md_def}
The \textnormal{\textbf{matching distance}}, $d_{M}$, between two 2-parameter persistence modules $\mathcal{M}$ and $\mathcal{N}$ is the supremum of the bottleneck distances between the persistence diagrams on corresponding lines of non-negative slope in the two modules. 
Precisely,
\[ d_{M} = \sup_{\ell} \{ d_B(\mathcal{D}(\mathcal{M}_\ell), \mathcal{D}(\mathcal{N}_\ell)) \cdot \mathrm{weight} \left(\mathrm{slope}(\ell) \right) \}, \]
where the supremum is over all lines of nonnegative slope and $\mathrm{weight}(m) = \frac{1}{\sqrt{1+q^2}}$, where $q = \max\left(m, \frac{1}{m}\right)$.

\end{definition} 

In the definition of the matching distance, a weight is assigned to each line $\ell$, which depends on the slope $\ell$. A line with slope $1$ gets the maximum weight, and the weight approaches zero as the slope approaches zero or infinity.
The weight is chosen such that if the interleaving distance between two persistence modules is 1, then the weighted bottleneck distance is at most 1 \cite{interleaving, lesnick}.\footnote{Prior to computing the matching distance, the persistence modules $\mathcal{M}$ and $\mathcal{N}$ are often \emph{normalized}. That is, the parameter axes for each module are rescaled so that the parameter values for all generators and relations occur in specified intervals on each axis. For details, see \cite{rivetPython}.}

\section{Computations and Analysis}

Our datasets are point clouds with simple structure depending on a few parameters. Adjusting these parameters allows us to change the size of topological features (namely, components and holes) that are captured by the persistence modules. 
We expect the matching distance between these persistence modules to reflect the parameter differences in the underlying datasets.

From each dataset, we construct a \textbf{density-Rips bifiltration}. That is, our two parameters are a density estimator and Euclidean distance. A density estimator $f$ is assigned to each point $p$ such that $f(p)$ is small if $p$ has many nearby neighbors and $f(p)$ is large if $p$ has no nearby neighbors; this causes points with more neighbors to appear before points with few neighbors in the density filtration.\footnote{We used a $k$-nearest-neighbor density estimator, but many other options are available.}
For any $r \in \mathbb{R}$, the Rips filtration is constructed on $f^{-1}(-\infty, r]$.
This produces a 2-parameter family of simplicial complexes with inclusion maps in the increasing directions of both density and distance.

We begin with very simple point-cloud datasets, each consisting of three points in the $xy$-plane.
We examine the matching distances while keeping two points fixed and moving the third point around the plane.
Three points are enough to produce a density-Rips bifiltration that is nontrivial in both density and distance.
Thus, we regard these datasets as the simplest datasets that allow us to study the effect of moving a single point on the matching distance.

The second collection of datasets consists of points sampled from two circles with radii $r$ and separated by distance $d$. These datasets have nontrivial structure in persistent homology of degree zero and one.
We generated many datasets with various values of $d$ and fixed $r$, investigating how varying $d$ affects the matching distance computed from zero-degree persistent homology modules.
We also fix $d$ and vary $r$, investigating the matching distance between first-degree persistent homology modules. 

We computed two-parameter persistent homology using RIVET, an interactive visualization software developed by Michael Lesnick and Matthew Wright \cite{rivet}.
A detailed description of RIVET and its algorithms appears in a comprehensive preprint by Lesnick and Wright \cite{lesnick2015interactive}.
Given point-cloud data, RIVET computes a 2-parameter persistence module, and then computes barcodes along linear slices of the persistence module. 
We approximated the matching distance from these barcodes using Python code written by Bryn Keller and Michael Lesnick \cite{rivetPython}, which uses a finite set of lines to approximate the supremeum in Definition \ref{md_def}.\footnote{The approximation algorithm requires us to specify the number of lines used in the approximation. This is achieved by specifying a \emph{grid-size} parameter, which determines the number of different slope and intercept values that the algorithm uses. For example, if  grid-size is $20$, then the algorithm uses $20$ slope values and $20$ intercept values to produce $400$ lines, computing the bottleneck distance along each. In our experiments, we found we found little difference in the approximated matching distance when grid-size was set to $20$ or to a large value, such as $50$, though the computation time increases according to the square of the grid-size value.}
We regard the matching distance as giving a measure of similarity between two point-cloud datasets.

\subsection{Three-Point Datasets}

For our first investigation, each data set consisted of two fixed points $A = (1,1)$ and $B = (6.1,1)$, and a third point $C_i$ in the $xy$-plane. 
Figure \ref{threedot} shows the locations of $A$ and $B$ (in red), as well as the locations of all $C_i$ (blue).
Since $A$ and $C_i$ are always the closest pair of points, we assigned to these points a density parameter of $1$, and then we assigned point $B$ a density parameter of $2$.

\begin{figure}[!htb] 
  \centering 
  \begin{tikzpicture} [scale=0.55]

\draw [line width=0.04cm] (0,0) -- (0,10);
\draw [line width=0.04cm] (0,0) -- (16,0);
\node [scale = 0.03cm] at (8,-1) {$x$};
\node [rotate = 90,scale = 0.03cm] at (-1,5) {$y$};

\node [scale = 0.03cm] at (-0.3,-0.3) {0};
\draw [line width=0.05cm] (0,3) -- (-0.2,3);
\node [line width=0.03cm] at (-0.5,3) {3};
\draw [line width=0.05cm] (0,6) -- (-0.2,6);
\node [line width=0.03cm] at (-0.5,6) {6};
\draw [line width=0.05cm] (0,9) -- (-0.2,9);
\node [line width=0.03cm] at (-0.5,9) {9};

\draw [line width=0.05cm] (5,0) -- (5,-0.2);
\node [line width=0.03cm] at (5,-0.5) {2};
\draw [line width=0.05cm] (10,0) -- (10,-0.2);
\node [line width=0.03cm] at (10,-0.5) {4};
\draw [line width=0.05cm] (15,0) -- (15,-0.2);
\node [line width=0.03cm] at (15,-0.5) {6};

\foreach \Point in {(2.5,2),(15.2,2)}{
    \node[red,scale=0.04cm] at \Point {\textbullet};
}

\node [scale = 0.03cm] at (2.5,1.5) {$A$};
\node [scale = 0.03cm] at (15,1.5) {$B$};
\draw [line width=0.05cm] (9,5) -- (12,5);
\node [scale = 0.03cm] at (13.7,5) {choices of $C_i$};

\foreach \x in {0, 0.46, 0.92, 1.38, 1.84, 2.30, 2.76, 3.22, 3.68, 4.14, 4.61, 5.07, 5.53, 5.99, 6.45, 6.91, 7.37, 7.83, 8.29, 8.75}{
 \draw[fill,blue] (\x,2) circle (2pt);
 \draw[fill,blue] (\x,2.5) circle (2pt);
 \draw[fill,blue] (\x,3) circle (2pt);
 \draw[fill,blue] (\x,3.5) circle (2pt);
 \draw[fill,blue] (\x,4) circle (2pt);
 \draw[fill,blue] (\x,4.5) circle (2pt);
 \draw[fill,blue] (\x,5) circle (2pt);
 \draw[fill,blue] (\x,5.5) circle (2pt);
 \draw[fill,blue] (\x,6) circle (2pt);
 \draw[fill,blue] (\x,6.5) circle (2pt);
 \draw[fill,blue] (\x,7) circle (2pt);
 \draw[fill,blue] (\x,7.5) circle (2pt);
 \draw[fill,blue] (\x,8) circle (2pt);
 \draw[fill,blue] (\x,8.5) circle (2pt);
 \draw[fill,blue] (\x,9) circle (2pt);
 \draw[fill,blue] (\x,9.5) circle (2pt);
 \draw[fill,blue] (\x,10) circle (2pt);
 }   

\end{tikzpicture}

\caption{Data sets used in three point exploration. The red points $A$ and $B$ are two fixed points, and the blue points are the possibilities for the point $C_i$.}
\label{threedot}
\end{figure}
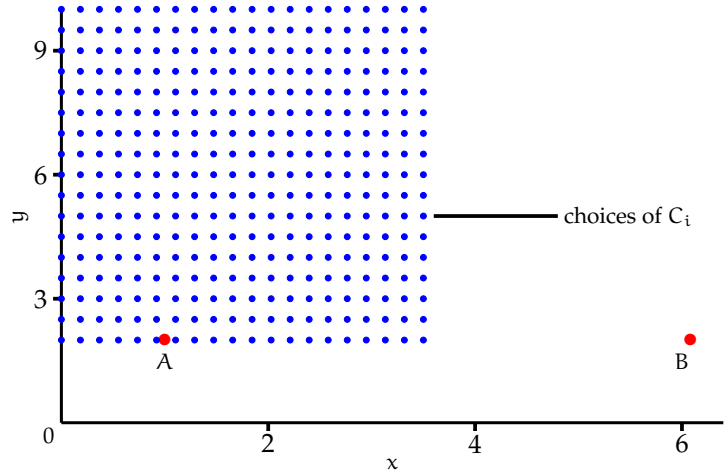

For concreteness, let $X_{r,s}$ denote a 3-point dataset where $C_i$ has coordinates $(r,s)$ . That is, $X_{r,s} = \{A, B, (r,s)\}$.
Let $X_{t,u} = \{A, B, (t,u)\}$ be another 3-point dataset.
We compute the matching distance between the 2-parameter persistence modules constructed from $X_{r,s}$ and $X_{t,u}$.

First, we fix $s = u = 3$. Figure \ref{distribution1} displays the matching distance between the 2-parameter persistence modules constructed from $X_{r,3}$ and $X_{t,3}$, for various choices of $r$ and $t$. The horizontal axis gives $r$, and the color of each curve represents the value of $t$. Specifically, $t$ ranges from $0$, colored dark green, to $3.3$, colored brown, and the step size is $.184$.

We observe that when $s$ and $u$ are small, such as $s = u = 3$, as shown in Figure \ref{distribution1}, there is a relatively linear increase in matching distances as the $t$ increases (above some threshold) with $r$ fixed. Some of the curves display a nonlinear region for small values of $t$; We note that these curves have $t$ smaller than $1$, which is the $x$-coordinate for $A$. 

\begin{figure}[ht]
    \centering
    \captionsetup{justification=centering}  
    \includegraphics[scale=0.5]{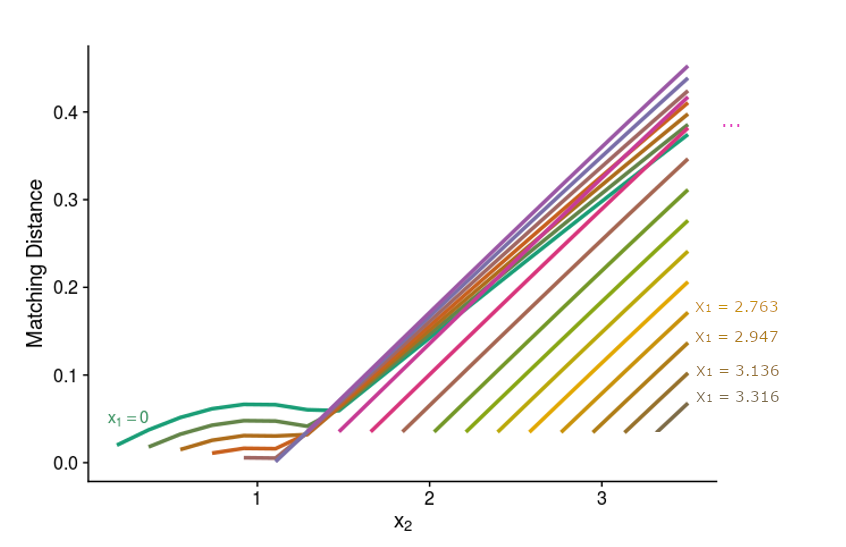}
    \caption{ Distribution of matching distance between $X_{r,3}$ and $X_{t,3}$. The horizontal axis gives the value of $t$, the color of the plot gives the value of $r$, and the vertical axis gives the matching
    distance between two datasets.}
    \label{distribution1}
\end{figure}

Next, we fix $s = u = 10$. Figure \ref{distribution1} displays the matching distance between the 2-parameter persistence modules constructed from $X_{r,10}$ and $X_{t,10}$, for various choices of $r$ and $t$.
Similarly, Figure \ref{distribution2} shows how the matching distance depends on $r$ and $t$ in this case.
It is clear that no linear trend in matching distances is present in these cases, as the matching distance increases faster as $r$ increases.
We note a prominent feature in Figure \ref{distribution2} is that the matching distance attains the value $0$; the following proposition explains why this occurs.

\begin{figure}[ht]
    \centering
    \captionsetup{justification=centering}  
    \includegraphics[scale=0.5]{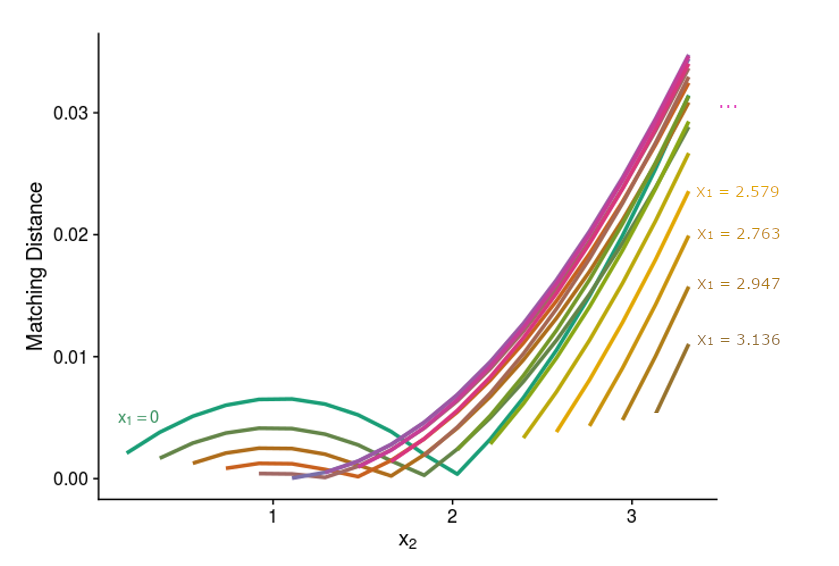}
    \caption{Distribution of matching distance between $X_{r,10}$ and $X_{t,10}$. The horizontal axis gives the value of $t$, the color of the plot gives the value of $r$, and the vertical axis gives the matching
    distance between two datasets.}
    \label{distribution2}
\end{figure}

\begin{proposition}\label{threepointproposition}
Suppose $A$ and $B$ are points on the $xy$-plane, such that the Euclidean distance between $A$ and $B$ is $d > 0$. Now suppose there are two points $C_1$ and $C_2$, both distance $r < d$ away from $A$, with the distance between $C_i$ and $B$ as $h_i > d$ for $i \in \{1,2\}$. Then the matching distance between the 2-parameter persistence modules constructed from the two point clouds $\{A, B, C_1\}$ and $\{A, B, C_2\}$ is $0$. (See Figure \ref{matchingdist})
\end{proposition}

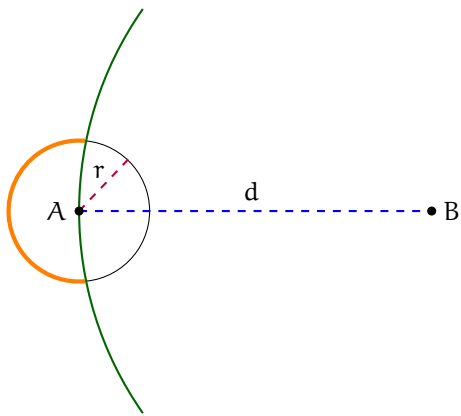
\begin{figure}[h]
    \centering
    \captionsetup{justification=centering}  
\begin{tikzpicture}[scale=0.75]
  \draw (1.25,2.2) circle (1.25cm);
  \draw[thick,dashed,blue] (1.25,2.2) -- (7.5,2.2);
  \draw[thick,dashed,purple] (1.25,2.2) -- (2.2,3.2);
  \node at (4.3,2.55) {$d$};
  \node at (1.6,2.9) {$r$};
  \draw [thick,black!60!green] ([shift=(215:6.25cm)]7.5,2.2) arc (215:145:6.25cm);
  \draw [ultra thick,orange] ([shift=(275:1.25cm)]1.25,2.2)arc (275:85:1.25cm);
  \filldraw (1.25,2.2) circle (2pt);
  \node at (0.85,2.2) {$A$};
  \filldraw (7.5,2.2) circle (2pt);
  \node at (7.86,2.2) {$B$};

\end{tikzpicture}
    \caption{Diagram of the points in Proposition \ref{threepointproposition}. Points $C_1$ and $C_2$ are located in the orange arc, at distance $r$ from $A$ and distance greater than $d$ from $B$.}
    \label{matchingdist}
\end{figure}

\begin{proof}
Since the distance between each $C_i$ and $A$ is smaller than the distance between each $C_i$ and $B$, the density parameter $1$ is assigned to $C_1$, $C_2$, and $A$, while point $B$ is assigned density parameter $2$. 
From such a point cloud $\{A, B, C_i\}$, we construct a  bifiltration as shown on the left side of Figure \ref{proof}, as we now explain, considering each density parameter value in turn.

\textit{Density 1}: When the distance parameter $\epsilon = 0$, only $C_1$ and $A$ appear, so we have two isolated points. When $\epsilon$ increases to $r$, $C_1$ and $A$ will connect to form an edge. As $\epsilon$ increases from $r$, the simplicial complex remains unchanged, since no other points exist to produce edges at this density.

\textit{Density 2}: When the scale parameter $\epsilon = 0$, all points $A,B,C_1$ in the point cloud appear, so we have three isolated points. When $\epsilon$ increases to $r$, $C_1$ and $A$ will form an edge. This results in one connected component and one isolated point.
When $\epsilon$ increases to $d$, an edge connects $B$ to $A$.
Since all points are connected at distance $d$, $H_0$ homology doesn't change as $\epsilon$ increases further.

Now consider the point cloud $\{A, B, C_2\}$. 
The bifiltration constructed from this point cloud is nearly the same as that described above; the only difference is the distance value $h_2$ at which the longest edge appears.
However, this edge does not connect any components that were not already connected at distance $d$, so no new (zeroth) homology appears at distance $h_2$.
Thus, we obtain topologically equivalent bifiltrations for the two data sets. This implies that the two 2-parameter persistence modules are the \textit{same}, and so barcodes along any linear slice of the two modules are also the same. 
Therefore, the matching distance between these modules is $0$.
\end{proof}

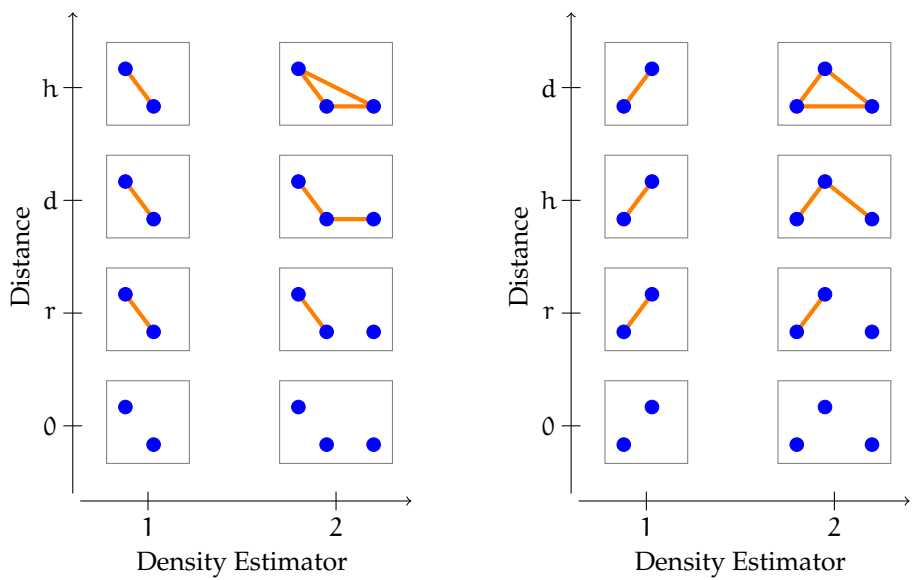
\begin{figure}[h]
  \centering
  \captionsetup{justification=centering}  
  \begin{tikzpicture}[scale=0.5]
    \draw [->] (0.2,0) -- (9,0);
    \draw [<-] (0,13) -- (0,0.2);
    \draw (0.25,2) -- (-0.25,2);
    \draw (0.25,5) -- (-0.25,5);
    \draw (0.25,8) -- (-0.25,8);
    \draw (0.25,11) -- (-0.25,11);
    \draw (2,-0.25) -- (2,0.25);
    \draw (7,-0.25) -- (7,0.25);
    \node at (4.5,-1.7) {Density Estimator};
    \node [rotate = 90] at (-1.4,6.5) {Distance};
    \node at (2,-0.7) {$1$};
    \node at (7,-0.7) {$2$};
    \node at (-0.6,2) {$0$};
    \node at (-0.6,5) {$r$};
    \node at (-0.6,8) {$d$};
    \node at (-0.6,11) {$h$};
    \draw [orange, ultra thick] (1.4,5.5) -- (2.15,4.5);
    \draw [orange, ultra thick] (1.4,8.5) -- (2.15, 7.5);
    \draw [orange, ultra thick] (1.4,11.5) -- (2.15, 10.5);
    \draw [orange, ultra thick] (6,5.5) -- (6.75,4.5);
    \draw [orange, ultra thick] (6,8.5) -- (6.75, 7.5);
    \draw [orange, ultra thick] (6,11.5) -- (6.75, 10.5);
    \draw [orange, ultra thick] (8,7.5) -- (6.75, 7.5);
    \draw [orange, ultra thick] (8,10.5) -- (6.75, 10.5);
    \draw [orange, ultra thick] (6,11.5) -- (8,10.5);
    \foreach \Point in {(1.4,2.5), (2.15,1.5), (1.4,5.5), (2.15,4.5), (1.4,8.5), (2.15, 7.5), (1.4,11.5), (2.15, 10.5), (6,2.5), (8,1.5), (6.75,1.5), (6,5.5), (8,4.5), (6.75,4.5), (6,8.5), (8,7.5), (6.75, 7.5), (6,11.5), (8,10.5), (6.75, 10.5)}{
        \filldraw [blue] \Point circle (.18);
    }
    \draw [gray] (0.9,1) rectangle ++(2.2,2.2);
    \draw [gray] (0.9,4) rectangle ++(2.2,2.2);
    \draw [gray] (0.9,7) rectangle ++(2.2,2.2);
    \draw [gray] (0.9,10) rectangle ++(2.2,2.2);
    
    \draw [gray] (5.5,1) rectangle ++(3.0,2.2);
    \draw [gray] (5.5,4) rectangle ++(3.0,2.2);
    \draw [gray] (5.5,7) rectangle ++(3.0,2.2);
    \draw [gray] (5.5,10) rectangle ++(3.0,2.2);
  \end{tikzpicture}
  \hspace{1cm}
  \begin{tikzpicture}[scale=0.5]
    \draw [->] (0.2,0) -- (9,0);
    \draw [<-] (0,13) -- (0,0.2);
    \draw (0.25,2) -- (-0.25,2);
    \draw (0.25,5) -- (-0.25,5);
    \draw (0.25,8) -- (-0.25,8);
    \draw (0.25,11) -- (-0.25,11);
    \draw (2,-0.25) -- (2,0.25);
    \draw (7,-0.25) -- (7,0.25);
    \node at (4.5,-1.7) {Density Estimator};
    \node [rotate = 90] at (-1.4,6.5) {Distance};
    \node at (2,-0.7) {$1$};
    \node at (7,-0.7) {$2$};
    \node at (-0.6,2) {$0$};
    \node at (-0.6,5) {$r$};
    \node at (-0.6,8) {$h$};
    \node at (-0.6,11) {$d$};
    \draw [orange, ultra thick] (1.4,4.5) -- (2.15,5.5);
    \draw [orange, ultra thick] (1.4,7.5) -- (2.15, 8.5);
    \draw [orange, ultra thick] (1.4,10.5) -- (2.15, 11.5);
    \draw [orange, ultra thick] (6,4.5) -- (6.75,5.5);
    \draw [orange, ultra thick] (6,7.5) -- (6.75, 8.5);
    \draw [orange, ultra thick] (6,10.5) -- (6.75, 11.5);
    \draw [orange, ultra thick] (8,7.5) -- (6.75, 8.5);
    \draw [orange, ultra thick] (8,10.5) -- (6.75, 11.5);
    \draw [orange, ultra thick] (6,10.5) -- (8,10.5);
    \foreach \Point in {(1.4,1.5), (2.15,2.5), (1.4,4.5), (2.15,5.5), (1.4,7.5), (2.15, 8.5), (1.4,10.5), (2.15, 11.5), (6,1.5), (8,1.5), (6.75,2.5), (6,4.5), (8,4.5), (6.75,5.5), (6,7.5), (8,7.5), (6.75, 8.5), (6,10.5), (8,10.5), (6.75, 11.5)}{
        \filldraw [blue] \Point circle (.18);
    }
    \draw [gray] (0.9,1) rectangle ++(2.2,2.2);
    \draw [gray] (0.9,4) rectangle ++(2.2,2.2);
    \draw [gray] (0.9,7) rectangle ++(2.2,2.2);
    \draw [gray] (0.9,10) rectangle ++(2.2,2.2);
    
    \draw [gray] (5.5,1) rectangle ++(3.0,2.2);
    \draw [gray] (5.5,4) rectangle ++(3.0,2.2);
    \draw [gray] (5.5,7) rectangle ++(3.0,2.2);
    \draw [gray] (5.5,10) rectangle ++(3.0,2.2);
  \end{tikzpicture}
  \caption{At left, the bifiltration constructed in the proof of Proposition \ref{threepointproposition}. At right, a similar bifiltration that results in a nonzero matching distance.}
    \label{proof}
\end{figure}

According to the proposition above, it would be easy to compute the probability of having matching distance of $0$ if $C_1$ and $C_2$ are randomly selected on circle $O_A$ and circle $O_B$ in Figure \ref{matchingdist}.
The following corollary gives such a probabilistic interpretation of the proposition; we leave the proof to the reader.

\begin{corollary} Suppose $A$ and $B$ are points on the $xy$-plane with distance between them $d > 0$. Let $O_A$ be a circle of radius $r<d$ centered at $A$, and $O_B$ a circle of radius $r<d$ centered at $B$.
Now suppose there are two points $C_1$ and $C_2$, which are randomly selected on the two circles. Then, the two point clouds $\{A, B, C_1\}$ and $\{A, B, C_2\}$ clouds have a matching distance of $0$ with probability 
$\left(1 - \frac{1}{\pi}\arccos^{-1}\left(\frac{r}{2d}\right) \right)^2$.
\end{corollary}

The following theorem gives a $n$-point generalization of the 3-point proposition. In the theorem, a dataset consists of $n$ vertices of a regular polygon in the $xy$-plane, as well as one additional point.

\begin{theorem} Suppose $A_1, A_2, \ldots, A_n$ are points on the $xy$-plane forming the vertices of a regular polygon with side length $d > 0$. Now suppose there are two points $C_1$ and $C_2$, both distance $r > 0$ (also $r < d)$ away from the $A_i$ that they are closest to, with the distance away from all other $A_i$ greater than $d$. 
Then, the matching distance between the 2-parameter persistence modules constructed from the two point clouds $\{A_1, A_2, \ldots, A_n, C_1\}$, and $\{A_1, A_2, \ldots, A_n, C_2\}$ will be $0$. 
\end{theorem}

\begin{proof} Similar to the proof for the proposition above, each point will be assigned a density parameter of $1$ or $2$. 
Since the distance between points $C_1, C_2$ and the respective $A_i$ that they are closest to is smaller than the distance between the $C_i$ and all other $A_i$, density $1$ is assigned to $C_1, C_2,$ and the respective $A_i$ that they are closest to (note that they could both be closest to the same $A_i$). All other $A_i$ are assigned a density of $2$. In the point clouds consisting of these points, we have a bifiltration for both density $1$ and density $2$.

In the bifiltration for density $1$, only $C_1$ and the $A_i$ that it is closest to appear, so the proof follows the same steps as the one for proposition. 

In the bifiltration for density $2$, all points in the point cloud appear (all $A_i$ and $C_1$). At $\epsilon = r$, $C_1$ and the $A_i$ it is closest to form an edge, but no other edges form between any points. As $\epsilon$ increases to $d$, all edges of the regular polygon appear.

Similarly, we have the exact same bifiltration for the point cloud consisting of $C_2$ and the $A_i$. Followed by the same reasonings in the proof for proposition, the matching distance between the two point clouds is $0$.
\end{proof}

\subsection{Two-Circle Datasets}

Our second investigation involved datasets consisting of points sampled from two circles with radii $r$, separated by distance $d$. 
For each dataset, two hundred points were selected randomly from a uniform distribution on each circle; Figure \ref{circle} displays an example.
Each point was assigned a density estimator defined as the distance to the $20$th nearest neighbor. 

\begin{figure}[h]
  \centering

  \begin{tikzpicture}
    \foreach \Point in {(0.93,-0.36),(-0.41,-0.91),(-0.93,0.38),(-0.46,-0.89),(-1,-0.06),(0.97,0.26),(0.9,-0.45),(0.52,0.85),(-0.99,-0.16),(-0.4,-0.92),(-0.76,-0.65),(-0.74,-0.67),(0.17,-0.98),(-0.45,-0.89),(0.25,-0.97),(0.18,-0.98),(-0.45,-0.89),(-0.46,-0.89),(-0.42,0.91),(-0.23,0.97),(-0.7,0.72),(0.17,-0.99),(-0.16,-0.99),(0.69,-0.72),(0.86,-0.51),(-0.77,-0.64),(-0.73,0.69),(0.03,1),(0.34,-0.94),(-0.99,-0.11),(-0.65,-0.76),(-0.56,-0.83),(0.87,0.5),(0.25,-0.97),(0.69,-0.72),(-0.35,0.94),(0.84,0.55),(-0.49,0.87),(-0.55,0.83),(0.2,-0.98),(0.94,-0.33),(0.92,-0.4),(0.92,0.39),(-0.52,-0.85),(0.81,0.59),(-0.99,-0.12),(0.22,0.97),(-0.48,0.88),(0.38,-0.93),(0.74,0.67),(-0.92,0.39),(0.21,-0.98),(-0.91,0.42),(0.43,0.9),(0.55,-0.83),(0.11,-0.99),(-0.5,-0.87),(0.62,0.79),(-0.42,-0.91),(0.23,0.97),(-0.58,0.82),(-0.79,0.61),(-0.56,-0.83),(0.75,0.66),(0.99,0.16),(0.61,0.79),(-0.71,0.71),(0.75,0.66),(0.99,0.1),(0.32,-0.95),(-0.42,0.91),(-0.13,-0.99),(-0.74,-0.67),(0.65,0.76),(0.92,-0.4),(-0.84,-0.54),(-0.69,0.72),(0.98,0.19),(-0.88,-0.47),(0.31,-0.95),(-0.62,-0.78),(-0.96,-0.29),(-0.83,-0.55),(0.85,0.52),(-0.21,-0.98),(-0.97,-0.24),(-0.07,1),(0.95,-0.31),(0.88,0.47),(0.6,0.8),(0.99,-0.12),(0.82,-0.57),(-0.15,0.99),(-0.98,-0.2),(0.23,-0.97),(-0.91,-0.42),(0.65,0.76),(0.62,-0.79),(0.06,1),(-0.56,0.83),(1,0.01),(0.25,0.97),(0.15,-0.99),(-0.68,-0.73),(0.95,-0.3),(-0.12,0.99),(-0.17,0.99),(0.87,0.49),(-0.91,-0.41),(0.54,0.84),(0.98,-0.19),(-0.55,-0.84),(-0.97,0.22),(0.97,-0.23),(-1,-0.07),(0.98,-0.2),(-0.01,1),(-1,-0.06),(0.76,-0.65),(0.89,-0.46),(-0.73,-0.69),(0.19,-0.98),(0.8,0.59),(-0.95,-0.32),(-0.31,-0.95),(-0.2,-0.98),(0.88,0.47),(1,-0.06),(-0.44,-0.9),(1,0.07),(0.83,-0.55),(-0.78,0.62),(-0.55,0.84),(-0.64,0.77),(0.6,-0.8),(0.65,-0.76),(-0.3,-0.95),(0.86,-0.51),(-0.78,0.62),(-0.17,0.99),(-0.92,-0.39),(-0.99,0.11),(0.75,0.66),(0.73,-0.69),(0.52,-0.85),(-0.03,-1),(-0.31,0.95),(-0.65,0.76),(0.44,0.9),(0.18,-0.98),(-0.99,0.11),(-0.97,0.24),(-0.93,-0.37),(-0.99,-0.12),(-0.82,-0.58),(0.21,-0.98),(-0.33,-0.94),(0.95,0.32),(0.74,-0.67),(0.55,-0.84),(0.96,-0.29),(0.87,0.5),(0.97,-0.26),(0.67,0.74),(-0.96,-0.28),(-0.62,0.79),(0.13,-0.99),(1,0.07),(-0.25,-0.97),(0.83,0.56),(0.93,0.38),(-0.08,1),(-0.93,0.37),(-0.52,-0.86),(-0.58,0.82),(0.35,0.94),(-0.1,-0.99),(0.93,-0.36),(0.74,0.67),(-0.85,-0.52),(0.76,-0.65),(-0.99,0.16),(0.09,1),(0.99,-0.1),(-0.88,-0.47),(-0.89,0.45),(0.28,-0.96),(-0.96,0.28),(-0.99,-0.16),(-0.72,-0.69),(-0.72,0.69),(0.97,0.26),(-0.69,0.72),(0.94,-0.33),(-0.89,-0.45),(-0.27,-0.96),(0.8,0.6),(-0.22,-0.97),(0.92,-0.4),(0.99,0.12),(3.45,0.89),(2.03,-0.22),(3.3,0.96),(3.69,-0.72),(3.61,-0.79),(2.67,-0.94),(2.75,0.97),(2.05,-0.3),(3.7,0.71),(3.24,0.97),(2.57,-0.9),(2.01,0.1),(2.01,-0.12),(2.18,-0.58),(3.98,-0.21),(2.82,-0.98),(2.25,-0.66),(2.3,-0.71),(2.01,0.14),(3.73,0.69),(2.14,0.51),(2.04,-0.29),(3.99,-0.14),(3.55,0.84),(2.74,-0.97),(3.12,-0.99),(3.84,0.54),(3.97,-0.26),(2.36,-0.77),(2.56,-0.9),(3.34,0.94),(2.22,-0.63),(3.83,0.56),(4,0.01),(2.46,0.84),(2.01,0.12),(3.61,-0.79),(2.11,-0.45),(3.67,-0.75),(3.98,-0.17),(3.71,-0.7),(2.58,0.91),(3.02,1),(3.94,0.35),(2.21,-0.62),(3.81,-0.58),(2.16,0.54),(2.01,0.15),(3.92,-0.4),(3.98,0.21),(3.93,0.38),(2.78,0.98),(3.95,0.33),(3.25,0.97),(3.84,0.54),(3.77,-0.64),(2.15,-0.53),(3.95,0.33),(3.96,-0.29),(3.98,-0.22),(2.02,0.18),(2.61,-0.92),(3.75,0.66),(2,-0.06),(3.96,-0.27),(2.04,0.29),(3.77,0.63),(2.38,0.78),(2,0.09),(2.06,0.34),(4,0.07),(3.83,0.56),(3.99,-0.14),(3.05,-1),(3.99,0.16),(4,-0.05),(2.96,-1),(2.73,-0.96),(2.48,-0.86),(3.84,-0.55),(3.89,0.45),(3.99,-0.13),(3.66,-0.75),(2,0),(2.01,-0.14),(2.05,-0.33),(2.24,0.65),(3.93,-0.37),(3.04,-1),(3.57,0.82),(3.7,0.72),(3.67,0.75),(3.57,-0.82),(3.37,-0.93),(3.93,-0.37),(3.69,-0.73),(3.93,0.37),(3.41,-0.91),(3.2,0.98),(3.96,-0.27),(3.3,0.95),(3.83,-0.55),(3.66,-0.75),(2.04,0.28),(3.95,-0.32),(2.62,-0.92),(2.99,1),(3.76,-0.64),(2.97,1),(3.52,0.85),(3.85,0.52),(3.84,-0.55),(3.67,-0.74),(3.18,0.98),(2.22,-0.62),(2.62,-0.92),(2.81,0.98),(2.02,0.21),(3.97,0.24),(3.43,0.9),(2.01,0.12),(2,0),(2.19,0.59),(3.08,-1),(2.13,0.49),(3.51,-0.86),(3.99,0.17),(2.17,-0.56),(3.87,0.49),(3.91,0.42),(3.72,-0.69),(2.06,0.34),(3.14,-0.99),(2.3,0.72),(2.63,0.93),(3.52,0.86),(2.44,-0.83),(2.72,-0.96),(3.44,0.9),(3.98,-0.19),(2,0.07),(3.21,-0.98),(3.69,-0.72),(3.63,0.78),(2,0.06),(4,-0.06),(3.79,-0.61),(3.59,0.81),(2.02,-0.19),(3.55,0.84),(2.06,0.34),(3.95,-0.32),(3.56,-0.83),(2.42,0.82),(2.47,-0.85),(3.58,0.81),(3.97,0.23),(3.65,0.76),(3.96,-0.29),(2.32,0.73),(3.67,-0.74),(2.13,0.49),(2.01,0.12),(3.11,-0.99),(2.34,-0.75),(2.55,-0.89),(3.53,0.85),(3.91,-0.4),(2.08,0.4),(3.71,0.7),(3.23,-0.97),(3.98,-0.19),(2,0.02),(3.45,0.89),(3.78,-0.63),(3.95,0.3),(3.92,-0.38),(3.99,0.12),(3.95,0.3),(3.51,-0.86),(2.29,0.7),(2.04,0.26),(2.11,0.46),(3.28,-0.96),(3.78,0.63),(3.26,0.97),(3.89,-0.46),(2.01,0.13),(3.34,-0.94),(3.26,0.97),(3.16,-0.99),(3.3,0.95),(2.33,0.74),(2.39,0.79),(2,-0.05),(3.52,-0.85),(3.98,-0.18),(3.9,0.44),(3.13,0.99),(3.72,-0.69)}{
      \filldraw \Point circle (0.7pt);
    }
    
    \draw [line width=0.02cm,black] (-2,-2) -- (-2,2);
    \draw [line width=0.02cm,black] (-2,-2) -- (5,-2);
    \draw [line width=0.02cm,black] (-2,2) -- (5,2);
    \draw [line width=0.02cm,black] (5,-2) -- (5,2);
    \draw [line width=0.008cm,gray] (-1,-2) -- (-1,2);
    \draw [line width=0.008cm,gray] (0,-2) -- (0,2);
    \draw [line width=0.008cm,gray] (1,-2) -- (1,2);
    \draw [line width=0.008cm,gray] (2,-2) -- (2,2);
    \draw [line width=0.008cm,gray] (3,-2) -- (3,2);
    \draw [line width=0.008cm,gray] (4,-2) -- (4,2);
    \draw [line width=0.008cm,gray] (-2,-1) -- (5,-1);
    \draw [line width=0.008cm,gray] (-2,0) -- (5,0);
    \draw [line width=0.008cm,gray] (-2,1) -- (5,1);
    
    \node [scale = 0.035cm] at (-1,-2.3) {$-1$};
    \node [scale = 0.035cm] at (0,-2.3) {$0$};
    \node [scale = 0.035cm] at (1,-2.3) {$1$};
    \node [scale = 0.035cm] at (2,-2.3) {$2$};
    \node [scale = 0.035cm] at (3,-2.3) {$3$};
    \node [scale = 0.035cm] at (4,-2.3) {$4$};
    \node [scale = 0.035cm] at (-2.3,-1) {$-1$};
    \node [scale = 0.035cm] at (-2.3,0) {$0$};
    \node [scale = 0.035cm] at (-2.3,1) {$1$};

  \end{tikzpicture}
  \caption{Example of circle dataset; the circles have radius $r=1$ and are separated by distance $d=1$.)}
\label{circle}
\end{figure}
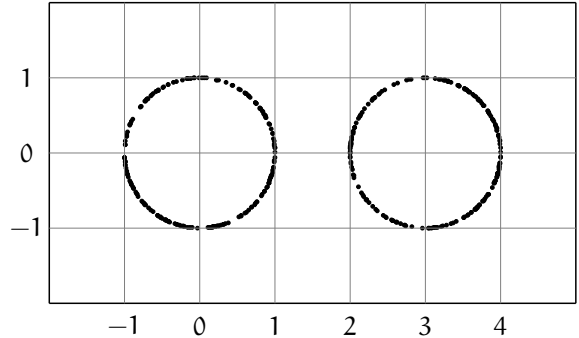

Suppose that we fix the radius $r$ and vary the separation distance $d$. Intuitively, the larger the change in separation distance, the more different we regard the point clouds.
Thus, we expect that a large change in $d$ will result in a large matching distance between the 2-parameter persistence modules constructed from the point clouds.

We generated $60$ datasets using $r=3$ and $d \in \{0.5, 1, 1.5, \ldots, 30\}$. 
We computed 2-parameter persistence modules for each dataset and computed the matching distance between each pair of modules.

Consider a pair of datasets, one with separation distance $d_1$ and the other with separation distance $d_2$.
Figure \ref{circledata} displays the matching distance between the persistence modules plotted against the \emph{difference} in separation distance $d_2 - d_1$. 
For the plot, we chose four representative values of $d_2$ from those listed above:  $d_2 \in \{10, 15, 22, 28\}$.
For each of these four values of $d_2$, we plot the matching distance for those datasets with separation distance $d_1$ such that the difference $d_2 - d_1$ in the range from $0.5$ to $d_2 - 0.5$.


Figure \ref{circledata} (left) shows the trend when $d_2$ is $10$ or $15$.
As the separation distance $d_2-d_1$ increases, the matching distance of the pairs of data sets increases initially, but then remains nearly constant.
In the region of increasing matching distance, we note three linear segments, with a jump between each.
Though the plots display nearly the same shape for $d_2=10$ and $d_2=15$, we note that the matching distance attains a higher value for $d_2 = 15$, indicating that the matching distance reveals a greater difference between the persistence modules when one dataset consists of circles that are farther apart.

When the separation distance $d_2$ is $22$ or $28$, as shown in Figure \ref{circledata} (right), the matching distance of that pair follows almost the same trend as in Figure \ref{circledata} (left). 
However, when $d_2 = 22$, we note some randomness, possibly due to sampling irregularities, when $d_2 - d_1$ is small.
Also, we note only two linearly increasing segments, with a single jump, when $d_2 = 28$.
Interestingly, the near-constant part of the plot is higher when $d_2$ is the smaller of the two values, which is contrary to what we observed in Figure \ref{circledata} (left).

\begin{figure}[h]
    \centering
    \captionsetup{justification=centering}  
    \includegraphics[scale=0.64]{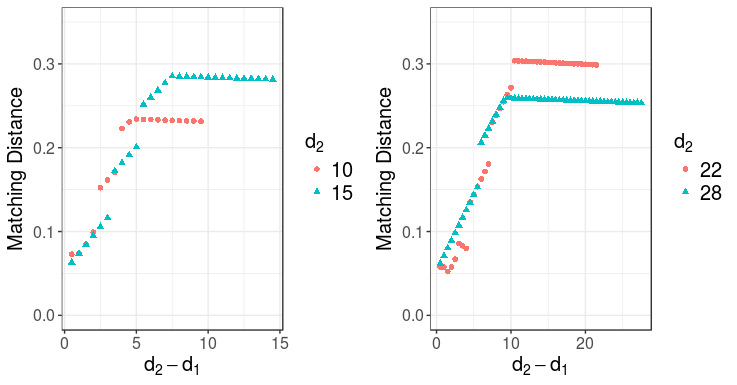}
    \caption{Distribution of matching distance for two-circle example}
    \label{circledata}
\end{figure}

In order to understand the structure observed in Figure \ref{circledata}, we looked into the barcodes involved in the matching distance calculation.
Recall that the matching distance between two 2-parameter persistence modules is the minimum weighted bottleneck distance between two barcodes, one from each persistence module.
RIVET provides us access to these barcodes.


For example, consider the pair of datasets for $d_1 = 8$ and $d_2 = 10$.
The matching distance between the persistence modules is plotted as one of the red dots in Figure \ref{circledata} (left).
The barcodes that realize this matching distance are shown (as persistence diagrams) in Figure \ref{circle_rivet}.
We compared these persistence diagrams to understand the matching distance between the persistence modules.
We observe that the persistence diagram for the dataset with separation distance $d_1 = 8$ has only three dots with finite coordinates\footnote{Our persistence diagrams for zero-degree persistent homology always have a dot at $(0,\infty)$ since there is one connected component that persists at all distance scales.} (i.e., the barcode has only three finite bars).
In comparison, the persistence diagram for the dataset with separation distance $d_2 = 10$ has five dots with finite coordinates (i.e., the barcode has only five finite bars). 
Given the small number of points, we manually matched the points and computed the matching distance. In this example, the two dots far from the diagonal are matched together in the optimal matching, and the distance between that pair is what gives the matching distance. By identifying how points are matched in persistence diagrams, we could better understand the structure in Figure \ref{circledata}.

\begin{figure}[h]
  \captionsetup{justification=centering}  
  \centering
  \begin{tikzpicture}[scale=0.25]
    \draw [<->,>=stealth] (15,0) -- (0,0) -- (0,15);
    \draw (0,0) -- (14,14);

    \foreach \Point in {(1.1,3),(3.1,4.1),(0.08,9.6)}{
        \filldraw [blue] \Point circle (0.3);
    }
    
    \foreach \Point in {(1.1,4.1),(2.2,5.8),(3.3,5.84),(4.4,5.6),(2.2,13.9)}{
        \filldraw [red] \Point circle (0.3);
    }
  \end{tikzpicture}

   \caption{Persistence diagrams involved in the matching distance calculation for $d_1=8$ (blue) and $d_2=10$ (red).}
    \label{circle_rivet}
\end{figure}
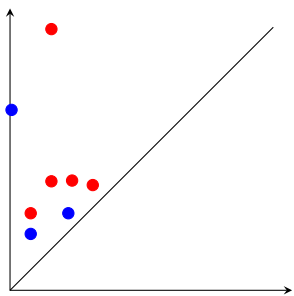

To understand the near-constant part of the plot, we looked at the persistence diagram that realize the matching distances between pairs of datasets.
Let $\mathcal{M}_d$ be the 2-parameter persistence module computed from the dataset with separation distance $d$.
In Figure \ref{circledata} (left), the near-constant portion of dots for $d_2 = 10$ extends from $d_2 - d_1 = 5$ (that is, $d_1 = 5$) to $d_2 - d_1 = 9.5$ (that is, $d_1 = 0.5$).
We examined the the calculation of the matching distance between $\mathcal{M}_{5}$ and $\mathcal{M}_{10}$, and also the matching distance between $\mathcal{M}_{0.5}$ and $\mathcal{M}_{10}$.
We found that in both matching distance calculations, the line $\ell$ that minimizes the bottleneck distance (thus realizing the matching distance) is the same. 
Let $\mathcal{D}^{\ell}_d$ be the persistence diagram obtained from $\mathcal{M}_d$ along line $\ell$.
We found that the barcode $\mathcal{D}^{\ell}_{10}$ has only one finite bar. When matching $\mathcal{D}^{\ell}_{10}$ to $\mathcal{D}^{\ell}_{0.5}$ or $\mathcal{D}^{\ell}_{5}$, this finite bar is always matched to the diagonal, and this gives the maximum distance between matched points. 
Thus, for the persistence modules that produce the near-constant portion of the plot, the matching distance is determined by the distance of the finite bar in $\mathcal{D}^{\ell}_{10}$ from the diagonal.

Since circles have nontrivial first homology, we performed a second experiment to investigate the effect of varying the circle radius on the matching distance between first persistent homology modules. 
Specifically, we generated data sets with circle radius $r \in \{0.2, 0.4, 0.6, \ldots, 6\}$ and fixed separation distance $d = 3$. 
We computed the 2-parameter persistence modules using first homology for each dataset, and compared them pointwise using the matching distance.
In this experiment, we observed patterns very similar to those displayed in Figure \ref{circledata}: as the difference in radius increases, the matching distance first increases, but then becomes near-constant. Furthermore, the near-constant portion is slightly decreasing as the difference in radius increases.

Lastly, we wanted to determine whether our two-circle experiments are sensitive to the presence of outliers, given that an important advantage of 2-parameter persistent homology is robustness against noise.
In our previous experiments, we sampled points precisely from the circles, but now we introduced some noise.
We added a small error to $20\%$ or $40\%$ of the data points, and re-computed the matching distances.
We found that the distribution of the matching distance for noisy data shares almost the same features as for the original data sets. 
We conclude that the matching distance still provides information about the change in separation distance or circle radius, even in the presence of noisy data. This confirms that two-parameter persistent homology is robust against outliers, which is one of the primary motivations for this study.

\section{Discussion and Future Research}

Our findings, though for very simple datasets, suggest that the matching distance can provide a notion of similarity for point-cloud data.
This research provides a step towards a deeper understanding of what matching distance reveals about the similarity or difference between point-cloud datasets. 
Moreover, this leads to further questions regarding how to quantify the similarity between geometric data.

In order to better understand our results in Figure $\ref{circledata}$, we would like to study why the jumps appear in the increasing segments in the plot of matching distances.
We would like to determine why the near-constant part of the plot is slightly decreasing as $d_2-d_1$ increases.
We are intrigued by the fact that the value of the matching distance along the near-constant part of the plot initially increases with $d_2$, but then decreases when $d_2$ gets sufficiently large.
In our experiments, the maximum matching distance occurs when $d_2 = 21$. 
We would like to study this further.

In this work, we obtained 1-parameter persistence modules by restricting 2-parameter persistence modules along lines of nonnegative slope.
We would like to explore the structures that exist along lines of negative slope, but this is algebraically complicated and would likely involve zigzag persistence, as discussed in \cite{zigzag}.

Furthermore, we would like to extend this research to the analysis of real-world data. 
To give one example, textual data such as Wikipedia articles can be converted to high-dimensional vectors --- e.g., using a word2vec algorithm.
We would like to use the matching distance to explore similarities between articles, and to compare collections of the article vectors with random point-cloud data.

\section*{Acknowledgements}

We would like to thank Professor Matthew Wright and Professor Matthew Richey for their guidance in this project. We would also like to thank the anonymous reviewer for many suggestions that greatly improved this paper. In addition we are grateful for the Collaborative Undergraduate Research and Inquiry (CURI) program for the generous support of undergraduate research at St.\ Olaf College. This work was supported by NSF DMS-1606967 and NSF DMS-1045015.


\bibliographystyle{siam}
\nocite{*}
{\small \bibliography{Matchingdistance} }

\end{document}